\documentclass[lettersize,journal]{IEEEtran}
\usepackage{amsmath,amsfonts}
\usepackage{algorithmic}
\usepackage{algorithm}
\usepackage{array}
\usepackage[caption=false,font=normalsize,labelfont=sf,textfont=sf]{subfig}
\usepackage{textcomp}
\usepackage{stfloats}
\usepackage{url}
\usepackage{verbatim}
\usepackage{graphicx}
\usepackage{cite}
\usepackage{amsthm}
\usepackage[dvipsnames]{xcolor}

\usepackage{graphicx}

\def\be{\begin{equation}}
	\def\ee{\end{equation}}
\def\ba{\begin{array}}
	\def\ea{\end{array}}
\usepackage[gen]{eurosym}
\usepackage{epstopdf}
\usepackage{dsfont}

\usepackage{multicol}

\usepackage{pgf}
\usepackage{tikz}
\usetikzlibrary{shapes,arrows}
\usetikzlibrary{positioning}
\usetikzlibrary{arrows,automata}
\usepackage{amsmath}
\usetikzlibrary{arrows,shadows,calc,decorations.shapes,decorations} \usepackage[underline=true,rounded corners=false]{pgf-umlsd}
\usepackage[american,cute inductors,smartlabels]{circuitikz}
\usepackage{mathrsfs}
\usepackage{graphicx,color}
\usepackage{array}
\usepackage{cite}
\usepackage{graphics} 	
\usepackage{latexsym}
\usepackage{booktabs}
\usepackage{amssymb}
\usepackage{amsmath,amssymb,amsfonts}
\DeclareMathOperator{\sgn}{sgn}

\usepackage{mathtools}

\newtheorem{remark}{Remark}

\newtheorem{definition}{Definition}

\newtheorem{theorem}{Theorem}

\newtheorem{lemma}{Lemma}

\newtheorem{assumption}{Assumption}
\newtheorem{corollary}{Corollary}
\usepackage{siunitx}	
\usepackage{balance}	

\DeclareMathAlphabet{\mathpzc}{OT1}{pzc}{m}{it}
\DeclareMathAlphabet{\mathcal}{OMS}{cmsy}{m}{n}

\newcommand{\R}{\mathbb{R}}

\DeclareMathAlphabet\mathbfcal{OMS}{cmsy}{b}{n}

\usepackage{color} 
\definecolor{darkgreen}{rgb}{0.0, 0.5, 0.0}

\usepackage{accents}
\newcommand{\ubar}[1]{\underaccent{\bar}{#1}}

\usepackage{tikz}
\usetikzlibrary{decorations.pathreplacing}

\definecolor{myLightGray}{RGB}{191,191,191}
\definecolor{myGray}{RGB}{160,160,160}
\definecolor{myDarkGray}{RGB}{144,144,144}
\definecolor{myDarkRed}{RGB}{167,114,115}
\definecolor{mylightRed}{RGB}{255,114,118}
\definecolor{myRed}{RGB}{255,58,70}
\definecolor{ForestGreen}{RGB}{188,255.0,77.5}
\definecolor{myGreen}{RGB}{0,255,71}

\usepackage{stackengine}

\begin{document}

\title{Stochastic Mean Field Game for Strategic Bidding of Consumers in Congested Distribution Networks}

\author{Amirreza~Silani, Simon~H.~Tindemans 
	\thanks{A.~Silani and  S.H.~Tindemans are with the 
		Department of Electrical Sustainable Energy, Delft University of Technology, The Netherlands 
		({\tt\small \{a.silani, s.h.tindemans\}@tudelft.nl})}
	\thanks{This research was supported by the GO-e project, which received funding from the MOOI subsidy programme by the Netherlands Ministry of Economic Affairs and Climate Policy and the Ministry of the Interior and Kingdom Relations, executed by the Netherlands Enterprise Agency.}
}

\markboth{Journal of \LaTeX\ Class Files,~Vol.~14, No.~8, August~2021}%
{Shell \MakeLowercase{\textit{et al.}}: A Sample Article Using IEEEtran.cls for IEEE Journals}

\maketitle

\begin{abstract}
The rapid increase of photovoltaic cells, batteries, and Electric Vehicles (EVs) in electric grids can result in congested distribution networks. An alternative to enhancing network capacity is a redispatch market, allowing Distribution System Operators (DSOs) to alleviate congested networks by asking energy consumers to change their consumption schedules. 
 However, energy consumers can anticipate the redispatch market outcomes and  strategically adjust their bids in the day-ahead market. This behaviour, known as increase-decrease gaming, can result in the exacerbation of congestion and enable energy consumers to gain windfall profits from the DSO. 
 In this paper, we consider a two-stage problem consisting of the day-ahead market (first stage) and redispatch market (second stage). Then, we model the increase-decrease game for large populations of energy consumers in power networks using a stochastic mean field game 
 approach. The agents (energy consumers) maximize their individual welfare in the day-ahead market with anticipation of the redispatch market. We show that all the agent strategies are ordered along their utilities and there exists a unique Nash equilibrium for this game. 
\end{abstract}


\section{Introduction}\label{sec:1}
There is a growing number of distributed energy resources in the electrical distribution grid, such as photovoltaic cells, batteries, and Electric Vehicles (EVs). The energy system  decarbonization  is essential for these loads; however, 
they introduce additional challenges for power networks, known as network congestion. 
There are various methods to alleviate the congestion issues including Local Flexibility Markets (LFMs) \cite{ref1,ref1.5,ref2}, Direct Load Control (DLC) \cite{ref2.0,ref2.1,ref2.2}, new forms of distribution tariffs \cite{ref2.3,ref2.4}, and infrastructure upgrade of distribution grid \cite{ref2.5}. Among them, 
we concentrate on the LFM proposals designed to deal with congestion management. 
Following the closure of the Day-Ahead (DA) market, LFM proposals usually consider energy consumption schedules.  In the LFM, when the forecast consumption schedules of flexible and inflexible loads lead to congestion issues, the Distribution System Operator (DSO) requests that energy consumers modify their consumption schedules, \textit{i.e.}, opening a local Redispatch (RD) market. 
To address the congestion problem, the DSO provides monetary compensation to the energy consumers  who decrease their consumption schedules in the redispatch market \cite{ref1,ref1.5,ref2}. 
However, energy consumers can anticipate the redispatch market outcomes and strategically adjust their bids in the day-ahead market, which results in the increase-decrease gaming \cite{lion}.  The increase-decrease game can lead to  the aggravation of congestion and enable energy consumers to gain
windfall profits from the DSO.  

Recently, there has been a growing research focus on  modeling and analysis of the increase-decrease game in analogous electrical energy markets (see for example \cite{lion,sarfati,sun,ntnu,Liu,Kohansal}). 
In \cite{lion}, an inconsistent power market design which combines a zonal spot market with a locational redispatch market is considered and it is shown that if a market-based solution is used, producers can anticipate the redispatch market and bid strategically in the spot markets. Indeed, the increase-decrease game is possible even when there exists no market power. A two-stage game is presented in \cite{sarfati} to analyze imperfect competition of producers in zonal power markets with a real-time and a day-ahead market, where strategic producers consider both markets when they bid in each market. In \cite{sun}, a profit decomposition method evaluating the effects of bidding strategies on the payoffs is proposed and the maximal markup is utilized as an index for extended Locational Marginal Price (eLMP) scheme to quantify market power.  
In \cite{ntnu}, a bi-level two-stage framework of the electricity market with a day-ahead market and a real-time redispatch market is proposed, where upper level is the strategic producer and lower level is the Transmission System Operator (TSO) follower.  
It is investigated in \cite{Kohansal} how a strategic market agent can maximize its profit when placing convergence bids and how such strategically submitted convergence bids can influence the price gaps between the day-ahead market and real-time market. 
Two auction methods consisting of marginal pricing and pay-as-bid pricing are proposed in \cite{Liu} to investigate generator's strategic bidding behaviors and it is shown that non-marginal generators will bid their true costs whilst the marginal one will bid above its costs under marginal pricing method and lower-cost generators are likely to bid their anticipated clearing price under a pay-as-bid pricing method.

Although, these papers investigate the increase-decrease game in different electrical energy markets, to the best of our knowledge, there exists no work in the literature which provide a suitable model and Nash equilibrium analysis for the increase-decrease game applied to distribution network congestion. Thus, in this paper, we study the increase-decrease game in distribution network congestion.

In competitive markets, the agents aim at maximizing their individual welfare with anticipation of the redispatch market (increase-decrease game) which may lead to so-called market failures where  congestion is not efficiently managed at the lowest cost. In LFM solutions for congestion management, there exist three types of market failures: price control (market power), fake schedules, and true modified schedules \cite{roman}. In the price control, few agents have control over a large amount of flexible loads and hold a monopoly over the congestion. In the fake schedules, the agents do not intend to fulfil their schedules since under the LFM mechanism, it will be paid to reduce their schedules anyways. In the true modified schedules, the agents bid the schedules that they really intends to fulfil; however, when they anticipate high energy demand in real-time, they increase their consumption to be paid to reduce it. Hence, the increase-decrease game can lead to market failures in congested distribution
networks and this vital problem has not been modeled and analyzed in the works in the literature (see for instance \cite{lion,sarfati,sun,ntnu,Liu,Kohansal}). In this paper, we mainly focus on the last problem and model and
analyze the increase-decrease game in congested distribution
networks. 

In LFM proposals for distribution network congestion management, typically the market has 
the multi-level (leader-follower)  structure of reverse Stackelberg game, where the followers move first and the leader's objective function includes a mapping from the decision space of the followers to the decision space of the leader \cite{Stackelberg1,Stackelberg2,Stackelberg3}. 
Therefore, we use the  
reverse Stackelberg game for modeling. 
Furthermore, since we do not want to consider the market power in the increase-decrease game, the population of energy consumers should be considered large, then the mean field game theory is used  for modeling. Indeed, 
mean field game is a methodology to investigate the multi-agent coordination problems which is classified into deterministic and stochastic approaches \cite{meanfield1,meanfield2}. In mean field game, each agent is affected by the statistical distribution of the population, and when the number of agents increases sufficiently, the contribution of each individual agent to the population distribution disappears \cite{sergio2,sergio,ma}. 
Thus, we exploit the stochastic mean field game theory and  multi-level structure of reverse Stackelberg game for modeling and analyzing the increase-decrease game applied to distribution network congestion for large populations of energy consumers. 

In this paper, we consider a two-stage problem consisting of the day-ahead market (first stage) and redispatch market (second stage). Then, we model the increase-decrease game for large populations of energy consumers in power networks using a stochastic mean field game 
 and reverse Stackelberg game approach. The energy consumers maximize their individual welfare in the day-ahead market with anticipation of the redispatch market. Then, we show that all the agent strategies are ordered along their utilities and there exists a unique Nash equilibrium for this game. The contributions of the paper can then be summarized as follows: (i) a two-stage problem consisting of the day-ahead market and redispatch market is considered; (ii)~the increase-decrease game for large populations of energy consumers in power networks is modeled using a stochastic mean field game 
approach; (iii) it is shown that all the agent strategies are ordered along their utilities; (iv)~the existence and the uniqueness of the Nash equilibrium for this game are theoretically proved. 

\subsection{Notation}
The set of positive (non-negative) real numbers is denoted by $\R_{>0}$ ($\R_{\geq0}$).
The $i$-th element of vector $x$ is denoted by $x_i$. 
Let  $[s]^{+} := s$ if $s> 0$, $0$ otherwise. 
$\mathbb{E}_{\nu}[\cdot]$ denotes the  mathematical expectation with respect to the random variable $\nu$. Let $f^{\prime}(x)$ be the derivative of $f(x)$. The deterministic variables are denoted by the small letters (for instance $\pi$) while the stochastic variables are denoted by the capital letters (for instance $\Pi$).

\section{Model description}\label{sec:2}
In this section, we describe the role of energy consumers and DSO in the day-ahead and redispatch market for congestion management. In market based congestion management, we have a two-stage problem consisting of the day-ahead market (first stage) and redispatch market (second stage). Fig.~\ref{fig1} illustrates this market model.

\begin{figure}[t]
	\begin{tikzpicture}[%
		every node/.style={
			font=\scriptsize,},]
	   
		\draw[->] (-0.87,0) -- (6.7,0);

		\foreach \x in {0,1.2,2.4,3.6,4.8,5.97}{
			\draw (\x cm,3pt) -- (\x cm,-3pt);
		}
\node[draw,
    circle,
    minimum size=1.05cm,
    fill=Goldenrod
] (DSO) at (-1.6,0.7){DSO};

\node[draw,
    circle,
    minimum size=1.05cm,
    fill=SeaGreen] (Market) at (-1.6,-0.7)  {Market};

\node[draw,
    circle,
    minimum size=1.05cm,
    fill=red!40] (Agents) at (-1.6,-2.1)  {Agents};

 \node[anchor=north] at (2.45,1.6) {\textbf{Stage 2}};
  
\node [draw,
    fill=Goldenrod,
    minimum width=1.75cm,
    minimum height=0.9cm,]  (Realize) at (2.45,0.75) {$\begin{array}{cc} d\leftarrow d_0 + D \\ \text{Redispatch?} \end{array}$};

\node [draw,
    fill=SeaGreen,
    minimum width=1.75cm,
    minimum height=0.9cm,]  (DAclearing) at (1.17,-0.7) {$\begin{array}{cc} 
 \text{DA clearing} \\ e^\mathrm{d}_{{i}}, \pi^\mathrm{d} \end{array}$};

 \node [draw,
    fill=SeaGreen,
    minimum width=1.75cm,
    minimum height=0.9cm,]  (RDclearing) at (4.8,-0.7) {$\begin{array}{cc} 
 \text{RD clearing} \\ e^\mathrm{r}_{{i}}, \pi^\mathrm{r} \end{array}$};

 \node[anchor=north] at (0,-1.25) {\textbf{Stage 1}};

 \node [draw,
    fill=red!40,
    minimum width=1.75cm,
    minimum height=0.9cm,]  (DAbids) at (0,-2.1) {$\begin{array}{cc} 
 \text{DA bids} \\ \hat{e}^\mathrm{d}_{{i}}, 
 \hat{\pi}^\mathrm{d}_{{i}} \end{array}$};

  \node [draw,
    fill=red!40,
    minimum width=1.75cm,
    minimum height=0.9cm,]  (RDbids) at (3.6,-2.1) {RD bids};

   \node [draw,
    fill=red!40,
    minimum width=1.4cm,
    minimum height=0.9cm,]  (Consume) at (6,-2.1) {Consume};

  \node[anchor=north] at (6.5,-0.1) {time};

     \draw[-stealth] (DAbids.east)  -- ++ (0.301,0)  node[near start,above]{} -- (DAclearing.south);

     \draw[-stealth] (DAclearing.north)  -- ++ (0,0.985)  node[near start,above]{} -- (Realize.west);

     \draw[-stealth] (RDbids.east)  -- ++ (0.322,0)  node[near start,above]{} -- (RDclearing.south);

     \draw[-stealth] (RDclearing.east)  -- ++ (0.306,0)  node[near start,above]{} -- (Consume.north);
     
  \draw[-stealth] (Realize.south) -- ++  (0,-0.99) node[midway,right]{Yes} 

  (Realize.south) -- ++ (0,-2.39) --  (RDbids.west);

     \draw[-stealth] (Realize.east)  -- ++ (2.56,0)  node[near start,above]{No} -- (Consume.north);
	
		\fill[myGreen] (-0.87,-2.85) rectangle (2.4,-3);
		\fill[gray] (2.4,-2.85) rectangle (5.34,-3);
	
		\fill[myRed] (5.34,-2.85) rectangle (6.7,-3);

		\draw[decorate,decoration={brace,amplitude=5pt}] (2.4,-3.2) -- (-0.87,-3.2)
		node[anchor=south,midway,below=4pt] {\textbf{DA}};
		\draw[decorate,decoration={brace,amplitude=5pt}] (5.4,-3.2) -- (2.4,-3.2)
		node[anchor=south,midway,below=4pt] {\textbf{RD}};
		\draw[decorate,decoration={brace,amplitude=5pt}] (6.7,-3.2) -- (5.4,-3.2) 
		node[anchor=south,midway,below=4pt] {\textbf{Real-time}};

	\end{tikzpicture}
	\centering
	\caption{Market model.}
	\label{fig1}
\end{figure}
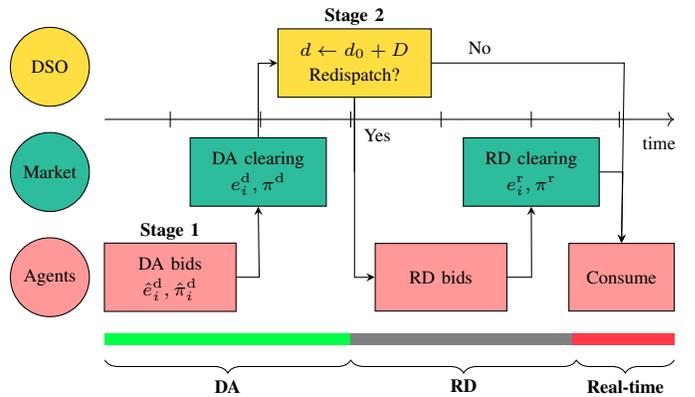

We consider the problem of energy consumption for  energy consumers (agents) with flexible loads. Indeed, each energy consumer has a utility (\textit{i.e.}, the intrinsic value of energy) to consume energy for its flexible loads.  We consider the set of agents $\mathcal{N}:=\{1,2,\dots,\infty\}$. 
For notational simplicity, we consider the problem for one time slot, but the same model is trivially applied to multiple (non-coupled) time slots. 
For each $i\in\mathcal{N}$, let $\hat{e}^\mathrm{d}_{{i}}$ and $\hat{\pi}^\mathrm{d}_{{i}}$ be the bid consumption schedule and bid price, $e^\mathrm{d}_{{i}}$ and $\pi^\mathrm{d}$ be the purchased consumption schedule and  clearing price (\textit{i.e.}, determined by pay-as-cleared rule) in the day-ahead market, $e^\mathrm{r}_{{i}}$ and $\pi^\mathrm{r}$ be the consumption reduction and price in the redispatch market (when congestion occurs) for its flexible loads, respectively, $u_{i}\in\R_{>0}$ be the utility  
and $e^{\mathrm{max}}_{i}$ be the maximum energy consumption. Furthermore, 
let $D$ be the stochastic variable of  forecast inflexible load demand which is realized in (near) real-time,  $d_0$ be the anticipated forecast inflexible load demand (known to all) and $c$ be the network capacity.

Each agent $i\in\mathcal{N}$ has the knowledge of $u_i$ and $e^{\mathrm{max}}_{i}$.
Then, each agent $i$ bids the consumption schedule  $\hat{e}^\mathrm{d}_{{i}}$, with the price $\hat{\pi}^\mathrm{d}_{{i}}$ in the day-ahead market (Stage~1). Then, the day-ahead market price $\pi^\mathrm{d}$ is cleared and agent $i$ purchases the consumption schedule  $e^\mathrm{d}_{{i}}=\hat{e}^\mathrm{d}_{{i}}$ if $\hat{\pi}^\mathrm{d}_{{i}}>\pi^\mathrm{d}$,  while it does not if $\hat{\pi}^\mathrm{d}_{{i}}<\pi^\mathrm{d}$, \textit{i.e.}, $e^\mathrm{d}_{{i}}=0$.  
Next, the DSO collects the schedules 
in the day-ahead market and based on the realization of $D$ if a congestion problem is detected, DSO asks the agents to reduce their consumption in the redispatch market (Stage~2).  
The agents who reduce their consumption in the redispatch market will be paid by the DSO. We note that the agents can either not trade on the day-ahead market, or trade.  
In the latter case, the agents can only trade on the day-ahead market or trade on both day-ahead and redispatch market. 

This model has the multi-level (leader-follower)  structure of reverse Stackelberg game, where the followers move first and the leader's objective function includes a mapping from the decision space of the followers to the decision space of the leader \cite{Stackelberg1}.  More precisely, the followers (energy consumers) first bid their consumption schedules and prices,  
then the day-ahead market price is cleared  and some followers  
trade and some do not trade on the day-ahead market. Then, the leader (DSO) collects the consumption schedules in the day-ahead market and opens the redispatch market based on the realization of $D$ if a congestion problem is detected. Next, the followers bid their consumption reductions, then the redispatch market price is cleared (this stage is trivial and considered implicit).  
Finally, some followers  
 trade and some do not trade on the redispatch market and the leader will pay to the followers who trade. Notably, the redispatch market bid, clearing and real-time consumption stages are implicit decisions or obvious problems. Indeed, two main stages are the day-ahead market bid (Stage~1) and opening the redispatch market or not (Stage~2).  We also note that the leader's costs for redispatching depend on the strategies of followers in both day-ahead and redispatch market. 

In this two-stage problem, the agents can anticipate the redispatch market and bid strategically in the day-ahead market to maximize their individual welfare which leads to the
increase-decrease gaming \cite{lion}. Indeed, the agents can anticipate when the congestion may occur in the network. 
Then, they modify their consumption schedules in the day-ahead market such that they will be paid by the DSO to reduce it in the redispatch market. 

Moreover, since we do not aim at taking into account the market power in the increase-decrease game applied to distribution network congestion, the population of energy consumers should be considered large. Therefore, we exploit the mean field game theory for modeling and analyzing the increase-decrease game, where each agent is affected by the statistical distribution of the population, and when the number of agents increases sufficiently, the contribution of each individual agent to the population distribution disappears \cite{sergio2,sergio,ma}.

\section{{Mean field game appraoch}}\label{sec:3}
In this section, we model the increase-decrease game for large populations of energy consumers using a stochastic mean field game approach. 

\subsection{Individual agent strategy} 
Each agent $i\in\mathcal{N}$ maximizes its individual  
welfare and can anticipate when the congestion may occur in the network. It modifies its schedule in the day-ahead market to maximize its individual welfare. The 
welfare of agent~$i$ is its utility minus day-ahead cost, plus anticipated
redispatch revenue 
 given by 
\begin{equation}\label{eq3.5}
W_i = \text{Utility} - \text{DA Cost} + \text{RD Revenue}. 
\end{equation}

In the following lemma, we obtain the optimal second stage strategy provided the purchased consumption schedule in the day-ahead market and price in the redispatch market. 
\begin{lemma}\textbf{{(Optimal second stage strategy).}}\label{lemma1}
Let the purchased consumption schedule $e_i^\mathrm{d}$ in the day-ahead and price $\pi^\mathrm{r}$ in the redispatch market be given. Then, the optimal second stage strategy is given by \footnote{The optimal second stage strategy of agent $i\in\mathcal{N}$ is indeterminate if $\pi^\mathrm{r}=u_i$. Agent $i$ is indifferent in this case.}: 
\begin{equation}\label{eq2}
e_i^\mathrm{r \ast} =      \begin{cases}
	e_{i}^\mathrm{d}, &\quad\text{if}~\pi^\mathrm{r}>u_i\\
	0, &\quad\text{if}~\pi^\mathrm{r}<u_i. 
\end{cases}
\end{equation}
\end{lemma}
\begin{proof} 
 The agent $i$ solves the following optimization problem in the redispatch market:
\begin{equation}\label{eq3}
	\begin{split}
		&\max_{e^\mathrm{r}_{i}}  \big[(\pi^\mathrm{r}-u_i)e^\mathrm{r}_{i}\big] \\
		&~~\text{s.t.}~ 0\leq e^\mathrm{r}_{i}\leq e^\mathrm{d}_{i},
	\end{split}
\end{equation}
The optimization problem \eqref{eq3} is a constrainted linear program; thus, we can conclude that if $\pi^\mathrm{r}>u_i$ then $e_i^\mathrm{r \ast}=e_i^\mathrm{d}$ and if $\pi^\mathrm{r}<u_i$ then $e_i^\mathrm{r \ast}=0$. 
\end{proof}

\begin{lemma}\textbf{{(Expected welfare).}}\label{lemma2}
Let assume that the agents have rational behavior in the redispatch market and $e_i^\mathrm{d}$ be given. Then, the expected welfare of agent $i\in\mathcal{N}$ can be rewritten as
\begin{align}\nonumber
		\mathbb{E}_{\Pi^\mathrm{r}}[W_i\vert \pi^\mathrm{d}, \Pi^\mathrm{r}, e_i^\mathrm{d}]  = 
		&e_i^\mathrm{d}\big(u_i -\pi^\mathrm{d}+\mathbb{E}_{\Pi^\mathrm{r}}\big[[\Pi^\mathrm{r}-u_i]^{+}\big]\big)\\ \label{eq4}
		:= & e_i^\mathrm{d} g(u_i).  
\end{align}
\end{lemma}

\begin{proof}
According to Lemma~\ref{lemma1} (and by assuming that the agents have rational behavior in the redispatch market), the expected  welfare of agent $i\in\mathcal{N}$  can be given by
\begin{align}\nonumber
	\mathbb{E}_{\Pi^\mathrm{r}}[W_i\vert \pi^\mathrm{d}, \Pi^\mathrm{r}, e_i^\mathrm{d}] = & \mathbb{E}_{\Pi^\mathrm{r}}[w_i\vert \pi^\mathrm{d}, \Pi^\mathrm{r}, e_i^\mathrm{d}, e_i^\mathrm{r \ast}] \\ \label{eq5}
	=& e_i^\mathrm{d} (u_i -\pi^\mathrm{d})+\mathbb{E}_{\Pi^\mathrm{r}}[e_i^\mathrm{r \ast}(\Pi^\mathrm{r}-u_i)],
\end{align}
where $e_i^\mathrm{r \ast}$ is given in \eqref{eq2}. Therefore, 
 by replacing  \eqref{eq2} in \eqref{eq5},  we can rewrite \eqref{eq5} as \eqref{eq4}. 
\end{proof}

\begin{lemma}\textbf{{(Optimal schedule in day-ahead market).}}\label{lemma3}
	The mapping $u\mapsto g(u)$ given in \eqref{eq4} is monotone increasing. Moreover, for each agent $i\in\mathcal{N}$, the optimal day-ahead schedule $e_i^\mathrm{d \ast}$ is determined based on $\sgn(g(u_i))$ as follows:
	\begin{equation}\label{eq8}
		e_i^\mathrm{d \ast} =  \begin{cases}
			e_i^\mathrm{\mathrm{max}}, &\quad\text{if}~g(u_i)>0\\
			0, &\quad\text{if}~g(u_i)<0 \\
               \text{indifferent,} &\quad\text{otherwise.}
                                    \\
		\end{cases}
	\end{equation}
\end{lemma}
\begin{proof}
Let $u_j\geq u_i$, with $i,j\in\mathcal{N}$. Then, 	the mapping $u\mapsto g(u)$  is monotone increasing if 
\begin{equation}\label{eq9}
	g(u_j) - g(u_i) \geq 0.
\end{equation}
Let $\rho_{\Pi^\mathrm{r}}$ be the probability density function of $\Pi^\mathrm{r}$. Then, we have
\begin{align}\nonumber
	g(u_j) - g(u_i) = & u_j - u_i + \mathbb{E}_{\Pi^\mathrm{r}}\big[[\Pi^\mathrm{r}-u_j]^{+}\big] \\ \nonumber
	& -\mathbb{E}_{\Pi^\mathrm{r}}\big[[\Pi^\mathrm{r}-u_i]^{+}\big]
	 \\ \nonumber
= & u_j - u_i + \int_{u_j}^{\infty} \rho_{\Pi^\mathrm{r}}(\sigma)\big(\sigma-u_j\big) d\sigma \\ \nonumber
 & -\int_{u_i}^{\infty} \rho_{\Pi^\mathrm{r}}(\sigma)\big(\sigma-u_i\big) d\sigma \\ \nonumber
= & (u_j - u_i)\big(1-\int_{u_j}^{\infty} \rho_{\Pi^\mathrm{r}}(\sigma) d\sigma\big) \\ \nonumber
& -\int_{u_i}^{u_j} \rho_{\Pi^\mathrm{r}}(\sigma)\big(\sigma-u_i\big) d\sigma \\ \nonumber
\geq & (u_j - u_i)\big(1-\int_{u_i}^{\infty} \rho_{\Pi^\mathrm{r}}(\sigma) d\sigma\big) \\
\geq & 0
\end{align}
 The first inequality holds using the Hölder's inequality \cite{holder} and the seconds inequality holds since $0\leq \int_{u_i}^{\infty} \rho_{\Pi^\mathrm{r}}(\sigma) d\sigma \leq 1$ and  $u_j\geq u_i$. Therefore, we can conclude 
that the mapping $u\mapsto g(u)$ given in \eqref{eq4} is monotone increasing.

In order to obtain the optimal day-ahead schedule, each agent $i\in\mathcal{N}$ should solve the following optimization problem
\begin{equation}\label{eq13}
	\begin{split}
		&\max_{e^\mathrm{d}_{i}} \mathbb{E}_{\Pi^\mathrm{r}}[ W_i\vert \pi^\mathrm{d}, \Pi^\mathrm{r}, e_i^\mathrm{d}]  \\
		&~~\text{s.t.}~0\leq e^\mathrm{d}_{i}\leq e^{\mathrm{max}}_{i}.
	\end{split}
\end{equation}
 According to Lemma~\ref{lemma2}, the expected welfare of agent $i\in\mathcal{N}$ is given in \eqref{eq4}. Thus, the optimization problem \eqref{eq13} is a linear program, whose solution depends on $\sgn(g(u_i))$. Indeed, if $g(u_i)>0$, then the optimal solution to the problem \eqref{eq13} is the upper bound of $e^\mathrm{d}_{i}$, \textit{i.e.}, $e^\mathrm{d}_{i}=e^{\mathrm{max}}_{i}$, and if $g(u_i)< 0$, then the optimal solution to the problem \eqref{eq13} is the lower bound of $e^\mathrm{d}_{i}$, \textit{i.e.}, $e^\mathrm{d}_{i}=0$. Moreover, if $g(u_i)=0$, then the optimal solution to the problem \eqref{eq13} is indeterminate and agent $i$ is indifferent in this case since its expected welfare is always equal to zero for any day-ahead schedules. Hence, $e^\mathrm{d \ast}_{i}$ given in \eqref{eq8} is the optimal schedule for each agent $i\in\mathcal{N}$.  
\end{proof}

\begin{theorem}\textbf{{(Optimal bid in day-ahead market).}}\label{theorem1}
For each agent $i\in\mathcal{N}$, the optimal bid in the day-ahead market is given by 
\begin{align}\nonumber
	(\hat{e}^\mathrm{d}_{i},\hat{\pi}^\mathrm{d}_{i}) = & \big(e^{\mathrm{max}}_{i},u_i+\mathbb{E}_{\Pi^\mathrm{r}}\big[[\Pi^\mathrm{r}-u_i]^{+}\big]\big) \\ \label{eq14}
	= & \big(e^{\mathrm{max}}_{i},\mathbb{E}_{\Pi^\mathrm{r}}[\max(u_i,\Pi^\mathrm{r})]\big),
\end{align}
 which results in 
 the optimal schedule \eqref{eq8}.
\end{theorem}
\begin{proof}
Following Lemma~\ref{lemma3}, for a given clearing price $\pi^\mathrm{d}$, the optimal schedule $e^\mathrm{d \ast}_{i}$ for each agent $i\in\mathcal{N}$ is given by
	\begin{equation}\label{eq15}
	e_i^\mathrm{d \ast} =  \begin{cases}
		e_i^\mathrm{\mathrm{max}}, &\quad\text{if}~u_i+\mathbb{E}_{\Pi^\mathrm{r}}\big[[\Pi^\mathrm{r}-u_i]^{+}\big]>\pi^\mathrm{d}\\
  	0, &\quad\text{if}~u_i+\mathbb{E}_{\Pi^\mathrm{r}}\big[[\Pi^\mathrm{r}-u_i]^{+}\big]<\pi^\mathrm{d}. \\
		\text{indifferent}, &\quad\text{otherwise}. \\ 
	\end{cases}
\end{equation}
Now, consider that each agent bids 
the consumption schedule  $\hat{e}_i^\mathrm{d}=e^{\mathrm{max}}_{i}$, with the price $\hat{\pi}^\mathrm{d}_{i}=u_i+\mathbb{E}_{\Pi^\mathrm{r}}\big[[\Pi^\mathrm{r}-u_i]^{+}\big]=\mathbb{E}_{\Pi^\mathrm{r}}[\max(u_i,\Pi^\mathrm{r})]$ in the day-ahead market. Then, according to the bid \eqref{eq14} and pay-as-cleared pricing rule, 
agent $i$ purchases the consumption schedule  $e^{\mathrm{max}}_{i}$ if $\hat{\pi}^\mathrm{d}_{{i}}>\pi^\mathrm{d}$, \textit{i.e.}, $e^\mathrm{d}_{{i}}=e^{\mathrm{max}}_{i}$, while it does not if $\hat{\pi}^\mathrm{d}_{{i}}<\pi^\mathrm{d}$, \textit{i.e.}, $e^\mathrm{d}_{{i}}=0$. Moreover, agent $i$ is indifferent to purchase the consumption schedule or not if $\hat{\pi}^\mathrm{d}_{{i}}=\pi^\mathrm{d}$. Consequently, bidding \eqref{eq14} is equivalent to consider the optimal schedule \eqref{eq15} for a given clearing price $\pi^\mathrm{d}$, which concludes the proof. 
\end{proof}

\subsection{Ordering of agent strategies} 
In this subsection, we investigate the ordering of agent strategies based on their utilities. In the following definition, we define the concept of Nash equilibrium \cite{sergio2}.
\begin{definition}\textbf{{(Nash equilibrium).}}\label{Definition1}
The set of admissible strategies for all agents $i\in\mathcal{N}$ is called a Nash equilibrium if each agent cannot increase its welfare
by changing its own strategies while the strategies of other agents are fixed. 
\end{definition}

\begin{lemma}\textbf{{(Trading or not on the day-ahead and redispatch market).}}\label{lemma5}
	Each agent $i\in\mathcal{N}$ can either not trade on, only trade on the day-ahead market or trade on both day-ahead and redispatch market. 
\end{lemma}
\begin{proof}
	Following Theorem~\ref{theorem1}, the optimal bid \eqref{eq14} in the day-ahead market is equivalent to consider the optimal schedule \eqref{eq8} for a given clearing price $\pi^\mathrm{d}$. Also, according to Lemma~\ref{lemma3}, $g(u_i)$ given in \eqref{eq4} is monotone increasing and the optimal schedule \eqref{eq8} depends on $\sgn(g(u_i))$. We consider the following possible cases based on $\sgn(g(u_i))$:
	\begin{itemize}
		\item[(i)] $g(u_i)<0$. According to Lemma~\ref{lemma3}, the optimal schedule $e_i^\mathrm{d \ast}$ for each agent $i\in\mathcal{N}$ is equal to $0$, \textit{i.e.}, agent $i$ does not  trade on the day-ahead market. 
		\item[(ii)] $g(u_i)>0$. According to Lemma~\ref{lemma3}, the optimal schedule $e_i^\mathrm{d \ast}$ for each agent $i\in\mathcal{N}$ is equal to $e^\mathrm{max}_{{i}}$, \textit{i.e.}, agent $i$  trades on the day-ahead market. 
		 Follwoing Lemma~\ref{lemma1}, in the second stage, if $u_i<\pi^\mathrm{r}$ agent $i\in\mathcal{N}$ trades on both day-ahead and redispatch market; otherwise,  
		  it only trades on the day-ahead market. 
    \item[(ii)] $g(u_i)=0$. According to Lemma~\ref{lemma3}, the optimal schedule $e_i^\mathrm{d \ast}$ for each agent $i\in\mathcal{N}$ is indeterminate, \textit{i.e.}, agent $i$  is indifferent to trades on the day-ahead market or not. 
	\end{itemize}
\end{proof}

\begin{theorem}\textbf{{(Ordering of agent strategies for trading on both day-ahead and redispatch market).}}\label{theorem2}
Assume that there exists a Nash equilibrium for this game. Then, in the Nash equilibrium, there exist $u^\mathrm{d}$, $U^\mathrm{r}$ such that $u^\mathrm{d} \in [\ubar{u},\bar{u}]$, with $\ubar{u}:=\arg\inf_{u} g(u)\geq 0$, $\bar{u}:=\arg\inf_{u} g(u)>0$,  
and  for each realization of $u^\mathrm{r}\leftarrow U^\mathrm{r}$, $u^\mathrm{d}\leq u^\mathrm{r}$. Moreover, all the agent strategies are ordered along $u$, with $u^\mathrm{d}\leq u^\mathrm{r}$ such that for all $i\in\mathcal{N}$
\begin{itemize}
\item[(i)]  If $u_i< u^\mathrm{d}$, agent $i$ does not trade on the day-ahead and redispatch market;  
	
\item[(ii)] if $u^\mathrm{d}< u_i<u^{\mathrm{r}}$, agent $i$ trades on both day-ahead and redispatch market with its maximum energy consumption; 
	
\item[(iii)] if $u^\mathrm{r}< u_i$, agent $i$ only trades on the day-ahead market with its maximum energy consumption. 
\end{itemize}
\end{theorem}
\begin{proof}
Firstly, we show that in the Nash equilibrium, for all $u_i\leq u_j$, with $i,j\in\mathcal{N}$, if agent $i$ trades on the day-ahead market and is not an indifferent agent, then agent $j$ also trades on the day-ahead market (and redispatch market or not). 
Since the  bid in the Nash equilibrium is optimal, following Theorem~\ref{theorem1}, the optimal bid \eqref{eq14} is equivalent to consider the optimal schedule \eqref{eq8} for a given clearing price $\pi^\mathrm{d}$. 
Also, in analogy with Lemma~\ref{lemma3}, the mapping $u\mapsto g(u)$ given in \eqref{eq4} is monotone increasing. Then, for all  $u_i\leq u_j$, with $i,j\in\mathcal{N}$, we have $g(u_i)\leq g(u_j)$. Now, let $u_i\leq u_j$ and agent $i$ trades on the day-ahead market (and redispatch market or not and is not an indifferent agent). Then, according to \eqref{eq8} in Lemma~\ref{lemma3}, $e_i^\mathrm{d \ast}=e_i^\mathrm{\mathrm{max}}$ and $g(u_i)>0$. Since $g(u_i)\leq g(u_j)$, we have $g(u_j)>0$.  
Therefore, according to Lemma~\ref{lemma3}, in the Nash equilibrium, the optimal schedule is $e_j^\mathrm{d \ast}=e_j^\mathrm{\mathrm{max}}$ and agent $j$ also trades on the day-ahead market. 

Next, we show that in the Nash equilibrium, for all $u_i\leq u_j$, with $i,j\in\mathcal{N}$, if agent $i$  only trades on the day-ahead market (not redispatch market) and is not an indifferent agent, then agent $j$ also only trades on the day-ahead market (not redispatch market).  
According to Lemma~\ref{lemma3}, the mapping $u\mapsto g(u)$ given in \eqref{eq4} is monotone increasing. Then, for all  $u_i\leq u_j$, with $i,j\in\mathcal{N}$, we have $g(u_i)\leq g(u_j)$. Now, let $u_i\leq u_j$ and agent $i$ only trades on the day-ahead market. Then, according to \eqref{eq8}  in Lemma~\ref{lemma3}, $e_i^\mathrm{d \ast}=e_i^\mathrm{\mathrm{max}}$, $g(u_i)>0$ and as we showed above, agent $j$ trades on the day-ahead market. We know that agent $i$ does not trade on the redispatch market, then we have $e^\mathrm{r}_{i}=0$.  Since the second stage bid in the Nash equilibrium is optimal, we have $e^\mathrm{r \ast}_{i}=e^\mathrm{r}_{i}=0$. Then, following Lemma~\ref{lemma1}, we have $\pi^\mathrm{r}<u_i$. Since $u_i\leq u_j$, we have $\pi^\mathrm{r}<u_j$. Consequently, following Lemma~\ref{lemma1}, in the Nash equilibrium $e^\mathrm{r \ast}_{j}=e^\mathrm{r}_{j}=0$ and agent $j$ does not trade on the redispatch market. 
	
Therefore, as we showed above, in the Nash equilibrium, for all $u_i\leq u_j$, with $i,j\in\mathcal{N}$, if agent $i$ trades on the day-ahead market and is not an indifferent agent, then agent $j$ also trades on the day-ahead market and if agent $i$ only trades on the day-ahead market and is not an indifferent agent, then agent $j$ also only trades on the day-ahead market. Indeed, all the agent strategies are ordered along $u$. Now, consider the utility order  $u_1\leq u_2\leq \dots \leq u_n$, where without loss of generality the agents are labeled based on the order of $u_i$. In analogy with Lemma~\ref{lemma3}, the mapping $u\mapsto g(u)$ given in \eqref{eq4} is monotone increasing. Then, we have 
\begin{equation}\label{eq18}
g(u_1)\leq g(u_2)\leq \dots \leq g(u_n).
\end{equation}

Now, we define $u^\mathrm{d} \in [\ubar{u},\bar{u}]$, with $\ubar{u}:=\arg\inf_{u} g(u)\geq 0$ and $\bar{u}:=\arg\inf_{u} g(u)>0$.  
Indeed, the agents with $u_i\in[\ubar{u},\bar{u}]$ are  indifferent agents, \textit{i.e.}, $g(u_i)=0$. Therefore, if $u_i< u^\mathrm{d}$, we have $g(u_i)< 0$.  
Then,  according to \eqref{eq8}  in Lemma~\ref{lemma3}, 
in the Nash equilibrium, 
agent $i$ does not trade on the day-ahead market when $g(u_i)< 0$, \textit{i.e.}, the optimal schedule is $e^\mathrm{d \ast}_{i}=0$. 
However, if $u_i> u^\mathrm{d}$, then we have $g(u_i)> 0$.  
Then, according to \eqref{eq8} in Lemma~\ref{lemma3}, in the Nash equilibrium, agent $i$ trades on the day-ahead market when $g(u_i)>0$, \textit{i.e.}, the optimal schedule is $e^\mathrm{d \ast}_{i} = e_i^\mathrm{\mathrm{max}}$. 
Now, based on the inflexible load demand realization, we have the following cases for each realization: 
\begin{itemize}
\item[(i)]
There exists no congestion. Then, there is no redispatch market and the agents with  $u_i> u^\mathrm{d}$ only trades on the day-ahead market. 
Thus, we consider $u^\mathrm{r}=u^\mathrm{d}$. 
\item[(ii)] 
There exists a congestion. Then, we consider $U^\mathrm{r}=\Pi^\mathrm{r}$ and for each realization $\pi^\mathrm{r}\leftarrow\Pi^\mathrm{r}$, we consider $u^\mathrm{r}=\pi^\mathrm{r}$. Following Lemma~\ref{lemma1}, if $u^\mathrm{d}<u_i<u^\mathrm{r}=\pi^\mathrm{r}$, the optimal second stage bid is $e_i^\mathrm{r \ast}=e_i^\mathrm{d \ast}=e_i^\mathrm{\mathrm{max}}\neq 0$. Thus, if $u^\mathrm{d}< u_i<u^\mathrm{r}$ in the Nash equilibrium, agent $i$ trades on both day-ahead and redispatch market.  
Moreover, following Lemma~\ref{lemma1}, if $u_i> u^\mathrm{r}=\pi^\mathrm{r}$, the optimal second stage bid is $e_i^\mathrm{r \ast}= 0$. Thus, if 
$u^\mathrm{r}< u_i$ in the Nash equilibrium, agent $i$  only trades on the day-ahead market. 
\end{itemize}

\end{proof}
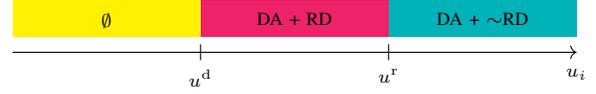
\begin{figure}[t]
	\begin{tikzpicture}[%
		every node/.style={
			font=\scriptsize,},
		]
	    
		\draw[->] (0,0) -- (7.5,0);
		
		\foreach \x in {2.5,5}{
			\draw (\x cm,3pt) -- (\x cm,-3pt);
		}

		\fill[yellow] (0,0.2) rectangle (2.5,0.7) ;
		\fill[WildStrawberry] (2.5,0.2) rectangle (5,0.7);
 
		\fill[BlueGreen] (5,0.2) rectangle (7.5,0.7); 
    \node[anchor=north] at (1.25,0.65) {$\emptyset$};
    \node[anchor=north] at (3.75,0.65) {DA + RD};
    \node[anchor=north] at (6.25,0.65) {DA + $\sim$RD};
   \node[anchor=north] at (2.5,-0.1) {$u^{\mathrm{d}}$};
    \node[anchor=north] at (5,-0.1) {$u^{\mathrm{r}}$};
    \node[anchor=north] at (7.5,-0.1) {$u_i$};
    
	\end{tikzpicture}
	\centering
	\caption{The agent outcomes based on their utilities in the Nash equilibrium.} 
	\label{fig1.1}
\end{figure}

We can conclude from Theorem~\ref{theorem2} that if there exists a Nash equilibrium for this game, then $u^\mathrm{d}$ and $u^\mathrm{r}$ can be used for summarizing all agent outcomes in the Nash equilibrium.
Fig.~\ref{fig1.1}  demonstrates the agent outcomes based on their utilities in the Nash equilibrium. 

\subsection{Stochastic mean field game}
In this subsection,  we model the increase-decrease game for large populations of energy consumers using a stochastic mean field game 
approach. 

In this approach, since we do not aim at considering the market power in the increase-decrease game applied to distribution network congestion, the population of agents should be considered large. Then, we use the mean filed game theory for modeling and analyzing
the increase-decrease game with the infinite number of agents. Now, let $\eta(u)$ be the total day-ahead flexible energy demand  probability density function with utility $u$. Then, define
\begin{equation}\label{phi}
	\phi(u):=\int_{u}^{\infty}\eta(w)dw, 
\end{equation}
 which denotes the total day-ahead flexible energy demand of the agents with the utility greater than or equal to $u$. Furthermore, let $f^{\mathrm{d}}(\cdot)$ be the supply function and $u^\mathrm{min}$, $u^\mathrm{max}$ be the minimum and maximum values for the utilities of the agents, respectively. 
 Now, we consider the following assumption on the functions $f^{\mathrm{d}}(\cdot)$,  $\phi(\cdot)$ and $\eta(\cdot)$. 

 \begin{assumption}\textbf{{(Conditions for $f^{\mathrm{d}}(\cdot)$, $\phi(\cdot)$, and $\eta(\cdot)$).}}\label{assumption1}
The supply function $f^{\mathrm{d}}(\cdot)$ is continuous and monotone  increasing,     
function $\phi(\cdot)$ is continuous, and if $u\in[u^\mathrm{min},u^\mathrm{max}]$, $\eta(u)>0$; otherwise,  $\eta(u)=0$. 
 \end{assumption}
 We note that according to \eqref{phi}, Assumption~\ref{assumption1} and the fact that the total day-ahead flexible energy demand  probability density function $\eta(u)$ is positive if $u\in[u^\mathrm{min},u^\mathrm{max}]$, we can show that the function $\phi(u)$ is strictly  decreasing if $u\in[u^\mathrm{min},u^\mathrm{max}]$ and continuous. Furthermore, Assumption~\ref{assumption1} implies that there exist no agents with utilities outside the  interval $[u^\mathrm{min},u^\mathrm{max}]$. 

\begin{lemma}\textbf{{(Optimal redispatching).}} \label{lemma13}
Let Assumption~\ref{assumption1} hold and  assume that there exists a Nash equilibrium for this game. Consider
\begin{equation}\label{eq28}
d_0+D+\phi(U^\mathrm{c})=c,
\end{equation}
where $U^\mathrm{c}$ is the congestion price. Then, we have the following cases:
\begin{itemize}
\item[(i)] The solution $U^\mathrm{c}=\phi^{-1}(c-d_0-D)$ to \eqref{eq28} exists. 
\item[(ii)] If $d_0+D+\phi(u)>c$ for all $u\in[u^\mathrm{min},u^\mathrm{max}]$ (the network is always congested), then the solution to \eqref{eq28} does not exist and we consider $U^\mathrm{c}=u^\mathrm{max}$. 
\item[(iii)] If $d_0+D+\phi(u)<c$ for all $u\in[u^\mathrm{min},u^\mathrm{max}]$ (the network is never congested), then the solution to \eqref{eq28} does not exist and we consider $U^\mathrm{c}=u^\mathrm{min}$. 
\end{itemize}
Moreover, in the Nash equilibrium, if $U^\mathrm{c}\leq u^\mathrm{d}$ there is no congestion and $U^\mathrm{r}=u^\mathrm{d}$; otherwise, there is a congestion, 
the optimal redispatching is given by \eqref{eq28} and $U^\mathrm{r}=U^\mathrm{c}$. 
  
\end{lemma}
\begin{proof}
According to \eqref{phi}, $\phi(U^\mathrm{c})$ is given by
\begin{equation}\label{eq29}
\phi(U^\mathrm{c}) = \int_{U^\mathrm{c}}^{\infty}\eta(w)dw,
\end{equation}
which implies the total  
flexible energy demand of the agents with the utility greater than or equal to $U^\mathrm{c}$.  Following Assumption~\ref{assumption1}, if $u\in[u^\mathrm{min},u^\mathrm{max}]$, the function $\phi(u)$ is strictly decreasing and continuous, then the function $\phi^{-1}(\cdot)$ exists and is continuous. Therefore, if there exists a $U^\mathrm{c}$ satisfying  \eqref{eq28}, $U^\mathrm{c}=\phi^{-1}(c-d_0-D)$ is the solution to \eqref{eq28}; otherwise, we consider $U^\mathrm{c}=u^\mathrm{max}$, if $d_0+D+\phi(u)>c$ for all $u\in[u^\mathrm{min},u^\mathrm{max}]$, \textit{i.e.}, the network is always congested and $U^\mathrm{c}=u^\mathrm{min}$ if $d_0+D+\phi(u)<c$ for all $u\in[u^\mathrm{min},u^\mathrm{max}]$, \textit{i.e.}, the network is never congested.  Following Theorem~\ref{theorem2}, in the Nash equilibrium,  all agents $i\in\mathcal{N}$ with $u_i< u^\mathrm{d}$ do not trade on the day-ahead market, with $u_i> u^\mathrm{d}$ trade on the day-ahead market and with $u_i> U^\mathrm{r}$ only trade on the day-ahead market (not redispatch market). Hence, according to Assumption~\ref{assumption1}, if $U^\mathrm{c}\leq u^\mathrm{d}$,  the total load demand (flexible and inflexible) is less than the network capacity, \textit{i.e.}, $d_0+D+\phi(u^\mathrm{d})\leq c$, then there is no congestion and $U^\mathrm{r}=u^\mathrm{d}$; otherwise, $d_0+D+\phi(u^\mathrm{d})>c$, then there is a congestion, 
$U^\mathrm{r}>u^\mathrm{d}$  and $\phi(U^\mathrm{r})$ is the total flexible energy demand of agents only trading on the day-ahead market (not redispatch market). 
Therefore, if the congestion problem occurs, 
the total flexible energy demand after redispatching the schedules becomes $\phi(U^\mathrm{r})$, \textit{i.e.},  the total day-ahead flexible energy demand of the agents with the utility greater than or equal to $U^\mathrm{r}$. 
 Now, we know that if congestion  occurs, the lowest cost solution for the DSO is that the total load demand (flexible and inflexible)  be equal to the  network capacity after redispatching the schedules, \textit{i.e.},  
\begin{equation}\label{eq32}
	d_0 + D + \phi(U^\mathrm{r}) = c.
\end{equation}
which leads to $U^\mathrm{r}=U^\mathrm{c}$. 
Indeed, in order to achieve the lowest cost solution for the DSO (optimal redispatching), $U^\mathrm{c}$ should satisfy \eqref{eq28} which concludes the proof. 
\end{proof}

\begin{lemma}\textbf{{(Redispatch price).}}\label{lemma11}
Assume that there exists a Nash equilibrium for this game. Then, in the Nash equilibrium,  
the redispatch price is given by
\begin{equation}\label{pir}
\Pi^\mathrm{r} =   \begin{cases}
		0, &\quad\text{if}~  U^\mathrm{c}\leq u^\mathrm{d}\\
		U^\mathrm{c}, &\quad\text{otherwise.} \\ 
	\end{cases}
\end{equation}
\end{lemma}
\begin{proof}
Following Theorem~\ref{theorem2},   
if $u^\mathrm{d}\leq u_i<U^\mathrm{r}$, agent $i\in\mathcal{N}$ trades on both day-ahead and redispatch market and if $u_i\geq U^\mathrm{r}$, it only trades on the day-ahead market. According to Lemma~\ref{lemma13}, if $U^\mathrm{c}\leq u^\mathrm{d}$, there is no congestion and no redispatch market. Consequently, the redispatch market price $\Pi^\mathrm{r}$ is equal to zero. However, following Lemma~\ref{lemma13}, if $U^\mathrm{c}> u^\mathrm{d}$, there is a congestion and a redispatch market. Therefore, according to Lemma~\ref{lemma1}  and Theorem~\ref{theorem2}, 
if $u^\mathrm{d}\leq u_i<\Pi^\mathrm{r}$, agent $i$ optimally trades on the redispatch market; otherwise, it does not and following Lemma~\ref{lemma13}, we have $U^\mathrm{r}=U^\mathrm{c}$.  Thus, due to the fact that the Nash equilibrium is locally optimal, we obtain   $\Pi^\mathrm{r}=U^\mathrm{c}$. 

\end{proof}

In the following, let 
$\rho_{D}$  be the probability density function of $D$. Then, in the following corollary, we show the bid price function as a function of the utility.  
\begin{corollary}\textbf{{(Bid price function).}}\label{corollary0}
	The bid price function is given by
	\begin{equation}\label{BidPrice}
		\hat{\pi}^\mathrm{d}(u|u^\mathrm{d}) = \mathbb{E}_{\Pi^\mathrm{r}}[\max(u,\Pi^\mathrm{r})], 
	\end{equation}
where 
	\begin{equation}
		\max(u,\Pi^\mathrm{r})=\begin{cases}
		u, &\quad\text{if}~  U^\mathrm{c}\leq \max(u^\mathrm{d},u)\\
		 U^\mathrm{c}, &\quad\text{otherwise.} \\ 
	\end{cases}
	\end{equation}
\end{corollary}
\begin{proof}
The bid price function \eqref{BidPrice} is obtained by applying Lemmas \ref{lemma13}  
and \ref{lemma11} in the bid price 
 given in Theorem~\ref{theorem1}.
\end{proof}

In the following lemma, the relationship between supply function and clearing price is shown. 

\begin{lemma}\textbf{{(Supply function and clearing price).}}\label{lemma10}
Let Assumption~\ref{assumption1} 
hold and assume that there exists a Nash equilibrium for this game. Consider the day-ahead clearing price $\pi^\mathrm{d}$ and supply function $f^{\mathrm{d}}(\cdot)$, then in the Nash equilibrium, we have
\begin{equation}\label{eq21}
	\pi^\mathrm{d}=f^{\mathrm{d}}(d_0 + \phi(u^\mathrm{d})).
\end{equation}
\end{lemma}
\begin{proof}
Following Theorem~\ref{theorem2}, in the Nash equilibrium, the agent strategies are ordered along $u$, with $u^\mathrm{d}\leq u^\mathrm{r}$ such that for all $i\in\mathcal{N}$ if $u_i < u^\mathrm{d}$, agent $i$ does not trade on the day-ahead market, if $u_i\geq u^\mathrm{d}$, agent $i$ trades on day-ahead market. 
Then, the total day-ahead flexible energy demand is given by
\begin{align}\nonumber
\text{Total day-ahead flexible demand}  = &  \underbrace{0\times\int_{0}^{u^\mathrm{d}}\eta(w)dw}_{\text{no agent trades}} \\ \nonumber 
 &  + 1\times\int_{u^\mathrm{d}}^{\infty}\eta(w)dw \\ \label{eq22}
	= & \int_{u^\mathrm{d}}^{\infty}\eta(w)dw =\phi(u^\mathrm{d}).
\end{align}
Then, the total anticipated load demand in the day-ahead market is sum of the anticipated inflexible load demand and total flexible load demand, \textit{i.e.},  $\phi(u^\mathrm{d}) + d_0$.  
Moreover, under Assumption~\ref{assumption1}, the supply function $f^{\mathrm{d}}(\cdot)$ and function $\phi(\cdot)$ are continuous. 
Then, the clearing price $\pi^\mathrm{d}$ should be equal to the price which is determined by the supply function $f^{\mathrm{d}}(\cdot)$ of energy demand level $d_0 + \phi(u^\mathrm{d})$, \textit{i.e.},
\begin{equation}\label{eq24}
 \pi^{\mathrm{d}}	= f^{\mathrm{d}}(d_0 + \phi(u^\mathrm{d})).
\end{equation}
which concludes the proof.
\end{proof}

\begin{remark}\textbf{(The problem of fixed bid price function).} \label{remark1}
    We note that if the bid price function \eqref{BidPrice} becomes a fixed value, 
the market may fail to determine which agents can trade on the day-ahead market and which cannot. Indeed,  the day-ahead price can be cleared at the flat part of the bid price function. Then, the load demand can be supplied only for some of the agents (not all of them) who bid the same price equal to the cleared day-ahead price and the market is unable to select which of these agents can trade on the day-ahead market or not. Therefore, in the following assumption, we assume that the congestion occurring is uncertain for all the agents, then we will show that this assumption leads to strictly increasing bid price function. Furthermore, according to Assumption~\ref{assumption1}, the probability density function $\eta(u)$ is positive if $u\in[u^\mathrm{min},u^\mathrm{max}]$; therefore, $U^\mathrm{c}=\phi^{-1}(c-d_0-D)$ is a continuous function, then  $\hat{\pi}^{\mathrm{d}}(u|u^\mathrm{d})$ is continuous and differentiable.
\end{remark}

Now, we consider the following assumption on congestion occurring.  
\begin{assumption}\textbf{(Uncertain congestion occurring).} \label{assumption2} 
      For all agents $i\in\mathcal{N}$, the congestion occurring is uncertain, \textit{i.e.}, $\mathrm{Pr}(c)\neq 1$. 
\end{assumption}

Then, in the following theorem, we investigate existence and uniqueness of the solution characterized by $u^\mathrm{d}$. 

\begin{theorem}\textbf{{(Existence and uniqueness of the solution characterized by $u^\mathrm{d}$).}} \label{theorem4}
	Let Assumptions~\ref{assumption1} and \ref{assumption2} hold  and  assume that there exists a Nash equilibrium for this game. 
	Then, in the Nash equilibrium, there exist a unique solution characterized by $u^\mathrm{d}$, defined by
	\begin{itemize}
		\item [(i)] If the supply function is below the bid price function at $u^\mathrm{min}$, \textit{i.e.}, 
		$f^{\mathrm{d}}(d_0 + \phi(u^\mathrm{min}))<\hat{\pi}^{\mathrm{d}}(u^\mathrm{min}|u^\mathrm{d})$, then the unique solution is defined by $u^\mathrm{d}=0$ 
  and all agents trade on the day-ahead market;
		\item[(ii)] if the supply function is above the bid price function at $u^\mathrm{max}$, \textit{i.e.}, 
		$f^{\mathrm{d}}(d_0 + \phi(u^\mathrm{max}))>\hat{\pi}^{\mathrm{d}}(u^\mathrm{max}|u^\mathrm{d})$,
		 then the unique solution is defined by 
   $u^\mathrm{d}=u^\mathrm{max}$ 
   and no agent trades on the day-ahead market;
		\item[(iii)] otherwise; the unique solution $u^\mathrm{d}$ is obtained from the intersection of the supply and bid price functions, \textit{i.e.},   $f^{\mathrm{d}}(d_0 + \phi(u^\mathrm{d}))=\hat{\pi}^{\mathrm{d}}(u^\mathrm{d}|u^\mathrm{d})$ or equivalently
			\begin{equation}\label{eq33}
			f^\mathrm{d}(d_0+\phi(u^\mathrm{d}))=u^\mathrm{d} + \mathbb{E}_{D}\big[[\phi^{-1}(c-d_0-D)-u^\mathrm{d}]^{+}\big].
		\end{equation}
		\end{itemize}
\end{theorem}
\begin{proof}
According to Theorem~\ref{theorem2}, in the Nash equilibrium, if $u_i< u^\mathrm{d}$, agent $i\in\mathcal{N}$ does not trade on the day-ahead market, while if $u^\mathrm{d}< u_i$, agent $i$ trades on the day-ahead with its maximum energy consumption.  
Following Lemma~\ref{lemma10}, in the Nash equilibrium, the day-ahead clearing price is given in \eqref{eq21}, \textit{i.e.,} $\pi^\mathrm{d}=f^{\mathrm{d}}(d_0 + \phi(u^\mathrm{d}))$. Then, according to Theorems~\ref{theorem1}, \ref{theorem2} and Lemma~\ref{lemma10}, agent $i$ does not trade on the day-ahead market 
 if and only if its bid price is less than the day-ahead clearing price 
  while agent $i$  trades on the day-ahead market 
if and only if its bid price is greater than the day-ahead clearing price, 
\textit{i.e.},
\begin{align}\label{iff0}
	u_i < & u^\mathrm{d} \Longleftrightarrow \pi^{\mathrm{d}}=f^{\mathrm{d}}(d_0 + \phi(u^\mathrm{d}))>\hat{\pi}^{\mathrm{d}}(u_i|u^\mathrm{d}) \\ \label{iff1}
	u_i > & u^\mathrm{d} \Longleftrightarrow	\pi^{\mathrm{d}}=f^{\mathrm{d}}(d_0 + \phi(u^\mathrm{d}))<\hat{\pi}^{\mathrm{d}}(u_i|u^\mathrm{d}).
\end{align}
According to Theorem~\ref{theorem2}, Lemmas \ref{lemma13}-\ref{lemma10} and Corollary~\ref{corollary0}, we have $u^\mathrm{d} \in [\ubar{u},\bar{u}]$, where $\ubar{u}:=\arg\inf_{u} \hat{\pi}^{\mathrm{d}}(u|u^\mathrm{d}) - f^{\mathrm{d}}(d_0 + \phi(u))\geq 0$
and $\bar{u} =  \arg\inf_{u} \hat{\pi}^{\mathrm{d}}(u|u^\mathrm{d}) - f^{\mathrm{d}}(d_0 + \phi(u))>0$.  
Now, we show that under Assumptions~\ref{assumption1} and \ref{assumption2}, the bid price function $\hat{\pi}^{\mathrm{d}}(u|u^\mathrm{d})$ is strictly increasing for all $u\in[u^\mathrm{min},u^\mathrm{max}]$. Also, since the probability density function $\eta(u)$ is positive if $u\in[u^\mathrm{min},u^\mathrm{max}]$, $U^\mathrm{c}=\phi^{-1}(c-d_0-D)$ is a continuous function, then  $\hat{\pi}^{\mathrm{d}}(u|u^\mathrm{d})$ is continuous and differentiable. Let $f_{u^\mathrm{c}}(u)$ be the probability density function of $U^\mathrm{c}$ given by \eqref{eq28}, then following Corollary~\ref{corollary0}, the bid price function can be given by
\begin{align}\nonumber
		\hat{\pi}^\mathrm{d}(u|u^\mathrm{d})  = & \mathbb{E}_{\Pi^\mathrm{r}}[\max(u,\Pi^\mathrm{r})]\\ \nonumber
  = & u \int_{-\infty}^{\max(u,u^\mathrm{d})}f_{u^\mathrm{c}}(u') du' \\ \label{BidPrice2}
  & + \int_{\max(u,u^\mathrm{d})}^{\infty} u' f_{u^\mathrm{c}}(u') du' 
	\end{align}
Then, we have the following cases:
\begin{itemize}
\item[(i)] If $u>u^{\mathrm{d}}$. Then, we have 
\begin{align}\nonumber
 \hat{\pi}^\mathrm{d}(u|u^\mathrm{d}) 
  = & u \int_{-\infty}^{u}f_{u^\mathrm{c}}(u') du' \\ \label{BidPrice3}
  & + \int_{u}^{\infty} u' f_{u^\mathrm{c}}(u') du'
\end{align}
The derivative of \eqref{BidPrice3} along $u$ is given by
\begin{align}\nonumber
\dfrac{\partial \hat{\pi}^\mathrm{d}(u|u^\mathrm{d})}{\partial u}
  = & \int_{-\infty}^{u}f_{u^\mathrm{c}}(u') du' \\ \nonumber
  = & 
  1- \int_{u}^{\infty}f_{u^\mathrm{c}}(u') du' \\ \label{BidPrice4}
  = & 1 - \mathrm{Pr}(U^\mathrm{c}>u)
\end{align}
Therefore, according to Assumption~\ref{assumption2} and due to the fact that $u>u^{\mathrm{d}}$, $\mathrm{Pr}(U^\mathrm{c}>u) \leq \mathrm{Pr}(U^\mathrm{c}>u^\mathrm{d}) =  \mathrm{Pr}(c)\neq 1$, then $\frac{\partial \hat{\pi}^\mathrm{d}(u|u^\mathrm{d})}{\partial u}>0$.  
\item[(ii)] If $u\leq u^{\mathrm{d}}$. Then, we have
\begin{align}\nonumber
 \hat{\pi}^\mathrm{d}(u|u^\mathrm{d})
  = & u \int_{-\infty}^{u^\mathrm{d}}f_{u^\mathrm{c}}(u') du' \\ \label{BidPrice5}
  & + \int_{u^\mathrm{d}}^{\infty} u' f_{u^\mathrm{c}}(u') du'
\end{align}
The derivative of \eqref{BidPrice5} along $u$ is given by
\begin{align}\nonumber
\dfrac{\partial \hat{\pi}^\mathrm{d}(u|u^\mathrm{d})}{\partial u}
  = & \int_{-\infty}^{u^\mathrm{d}}f_{u^\mathrm{c}}(u') du' \\ \nonumber
  = & 
  1- \int_{u^\mathrm{d}}^{\infty}f_{u^\mathrm{c}}(u') du' \\ \nonumber
  = & 1 - \mathrm{Pr}(U^\mathrm{c}>u^\mathrm{d}) \\ \label{BidPrice4}
  = &  1 - \mathrm{Pr}(c)
\end{align}
Therefore, according to Assumption~\ref{assumption2}, $\mathrm{Pr}(c)\neq 1$, then $\frac{\partial \hat{\pi}^\mathrm{d}(u|u^\mathrm{d})}{\partial u}>0$. 
\end{itemize}
Consequently, we can conclude from items~(i) and (ii) that the bid price function $\hat{\pi}^{\mathrm{d}}(u|u^\mathrm{d})$ is strictly increasing for all $u\in[u^\mathrm{min},u^\mathrm{max}]$. 
Furthermore, under 
Assumption~\ref{assumption1}, the supply function $f^{\mathrm{d}}(\cdot)$ is continuous and monotone  increasing and the function $\phi(\cdot)$ is continuous and strictly decreasing for all $u\in[u^\mathrm{min},u^\mathrm{max}]$.  
Then, according to Lemma~\ref{lemma3} and Assumptions~\ref{assumption1} and \ref{assumption2}, the mapping $g(u)=\hat{\pi}^{\mathrm{d}}(u|u^\mathrm{d}) - f^{\mathrm{d}}(d_0 + \phi(u))$  is  strictly increasing for all $u\in[u^\mathrm{min},u^\mathrm{max}]$. 
Consequently, if $f^{\mathrm{d}}(d_0 + \phi(u^\mathrm{min}))<\hat{\pi}^{\mathrm{d}}(u^\mathrm{min}|u^\mathrm{d})$, we obtain 
$u^\mathrm{d}=0$ since $u\geq 0$ and $\eta(u)=0$ for all $u<u^\mathrm{min}$. 
Then, according to \eqref{iff1},  
and due to the fact that $u^\mathrm{min}\geq 0$,  
all agents trade on the day-ahead market. If $f^{\mathrm{d}}(d_0 + \phi(u^\mathrm{max}))>\hat{\pi}^{\mathrm{d}}(u^\mathrm{max}|u^\mathrm{d})$, we obtain  
$u^\mathrm{d}=u^\mathrm{max}$ since  $\eta(u)=0$ for all $u>u^\mathrm{max}$. 
Then, according to \eqref{iff0} no agent trades on the day-ahead market.  
Otherwise, the supply and bid price functions has an intersection point. Now, according to \eqref{BidPrice2}, the bid price function at $u^\mathrm{d}$ can be expressed as 
\begin{align}\nonumber
 \hat{\pi}^\mathrm{d}(u^\mathrm{d}|u^\mathrm{d}) 
  = & u^\mathrm{d} \int_{-\infty}^{u^\mathrm{d}}f_{u^\mathrm{c}}(u') du' \\ \label{BidPrice7}
  & + \int_{u^\mathrm{d}}^{\infty} u' f_{u^\mathrm{c}}(u') du'.
\end{align}
The derivative of \eqref{BidPrice7} along $u^{\mathrm{d}}$ is given by
\begin{align}\nonumber
\dfrac{\partial \hat{\pi}^\mathrm{d}(u^\mathrm{d}|u^\mathrm{d})}{\partial u^\mathrm{d}}
  = & \int_{-\infty}^{u^\mathrm{d}}f_{u^\mathrm{c}}(u') du' \\ \nonumber
  = & 
  1- \int_{u^\mathrm{d}}^{\infty}f_{u^\mathrm{c}}(u') du' \\ \nonumber
  = & 1 - \mathrm{Pr}(U^\mathrm{c}>u^\mathrm{d}) \\ \label{BidPrice9}
  = & 1 - \mathrm{Pr}(c)
\end{align}
Therefore, according to Assumption~\ref{assumption2}, $\mathrm{Pr}(c)\neq 1$, then $\hat{\pi}^{\mathrm{d}}(u^\mathrm{d}|u^\mathrm{d})$  is strictly increasing for all $u^\mathrm{d}\in[u^\mathrm{min},u^\mathrm{max}]$. Hence, the uniqueness of the solution can be concluded from the strictly increasing property of  the mapping $g(u^\mathrm{d})=\hat{\pi}^{\mathrm{d}}(u^\mathrm{d}|u^\mathrm{d}) - f^{\mathrm{d}}(d_0 + \phi(u^\mathrm{d}))$. 

Moreover, following \eqref{eq28}, we obtain 
\begin{align} \nonumber
    f_{u^\mathrm{c}}(u) = & \rho_{D}(c-d_0-\phi(u)) \dfrac{\partial D}{\partial u} \\ \label{density}
    = & \rho_{D}(c-d_0-\phi(u)) \eta(u)
\end{align}
Then, by applying \eqref{density} to \eqref{BidPrice7}, we have 
\begin{align}\nonumber
 \hat{\pi}^\mathrm{d}(u^\mathrm{d}|u^\mathrm{d})
  = & u^\mathrm{d} \int_{-\infty}^{u^\mathrm{d}}\rho_{D}(c-d_0-\phi(u')) \eta(u') du' \\ \nonumber
  & + \int_{u^\mathrm{d}}^{\infty} u' \rho_{D}(c-d_0-\phi(u')) \eta(u') du' \\ \nonumber
  = & u^\mathrm{d} +\int_{u^\mathrm{d}}^{\infty}  \rho_{D}(c-d_0-\phi(u')) \eta(u')(u'\\ \nonumber
  & -u^{\mathrm{d}}) du' \\ \nonumber
  = & u^\mathrm{d} + \int_{c-d_0-\phi(u^\mathrm{d})}^{\infty}  \rho_{D}(\delta)\big(\phi^{-1}(c-d_0-\delta) \\ \nonumber 
  & -u^{\mathrm{d}}\big)d\delta \\ \label{BidPrice8}
  = & u^\mathrm{d} + \mathbb{E}_{D}\big[[\phi^{-1}(c-d_0-D)-u^\mathrm{d}]^{+}\big]. 
\end{align}
Thus,  $u^\mathrm{d}$ is obtained from solving \eqref{eq33}. 
\end{proof}

\begin{theorem}\textbf{{(Existence and uniqueness of a Nash equilibrium).}} \label{theorem5}
Let Assumptions~\ref{assumption1} and \ref{assumption2} hold. 
Then, the agent strategies are determined by $u^\mathrm{d}$ and $U^\mathrm{r}$ such that there exists a unique Nash equilibrium for this game.
\end{theorem}
\begin{proof}
Let $\tilde{u}^\mathrm{d}$ and $\tilde{U}^\mathrm{r}$ satisfy the conditions of Theorem~\ref{theorem2}, 
regardless of the existence of a Nash equilibrium for this game. 
 Then, we want to show that if the agents trade on the day-ahead and redispatch market as described in Theorem~\ref{theorem2} with the redispatch price given in Lemma~\ref{lemma11}, 
 the agent strategies are the Nash equilibrium. Thus, 
 let $\tilde{\Pi}^\mathrm{r}$ be given by \eqref{pir} with $\tilde{U}^\mathrm{r}$ and the agents trade on the day-ahead market as described in Theorem~\ref{theorem2} with $\tilde{u}^\mathrm{d}$ and $\tilde{U}^\mathrm{r}$, regardless of the existence of a Nash equilibrium for this game. Indeed, 
 for each realization of $\tilde{u}^\mathrm{r}\leftarrow \tilde{U}^\mathrm{r}$ and $\tilde{\pi}^\mathrm{r}\leftarrow \tilde{\Pi}^\mathrm{r}$,
 let the agents with  $u_i< \tilde{u}^\mathrm{d}$  not trade on the day-ahead market, with $\tilde{u}^\mathrm{d}< u_i<\tilde{u}^{\mathrm{r}}$ trade on both day-ahead and redispatch market with their maximum energy consumption and with $\tilde{u}^\mathrm{r}< u_i$ only trade on the day-ahead market with their maximum energy consumption. 
 Then, following Lemma~\ref{lemma1}, the second stage strategy is the optimal strategy \eqref{eq2} with $\tilde{\pi}^\mathrm{r}$. 
  Since only the agents with $\tilde{u}^\mathrm{d}< u_i$  trades on the day-ahead market, the day-ahead clearing price is given in \eqref{eq21} with $\tilde{u}^\mathrm{d}$, \textit{i.e.}, $\tilde{\pi}^\mathrm{d}=f^{\mathrm{d}}(d_0 + \phi(\tilde{u}^\mathrm{d}))$.  
Then, following Lemma~\ref{lemma3}, the day-ahead schedule of agent $i\in\mathcal{N}$ is the optimal schedule \eqref{eq8} with  $\tilde{\pi}^\mathrm{d}$ and $\tilde{\Pi}^\mathrm{r}$; thus, it is a Nash equilibrium for this game.  
Hence, following Theorem~\ref{theorem1}, 
 if all agents bid \eqref{eq14}, then there exists a Nash equilibrium for this game such that the agent strategies are determined by $u^\mathrm{d}=\tilde{u}^\mathrm{d}$ and $U^\mathrm{r}=\tilde{U}^\mathrm{r}$. According to Theorem~\ref{theorem4}, under Assumptions~\ref{assumption1} and \ref{assumption2}, in the Nash equilibrium,  there exist a unique solution characterized by $u^\mathrm{d}$. Therefore, we can conclude the uniqueness of the Nash equilibrium.  
\end{proof}

\subsection{Fixed day-ahead price and certain congestion occurring} \label{subsec:D}
In some applications, the day-ahead price  is considered fixed. Therefore,  
the supply function $f^{\mathrm{d}}(\cdot)$ is not
strictly increasing. In order to investigate the fixed day-ahead price, we consider the case that the supply curve can be flat for some intervals and the probability density function $\eta(u)$ is non-negative if $u\in[u^\mathrm{min},u^\mathrm{max}]$. Moreover, we consider that the congestion occurring can be certain, \textit{i.e.}, Assumption~\ref{assumption2} does not hold, then the bid price function $\hat{\pi}^{\mathrm{d}}(u|u^{\mathrm{d}})$ is monotone increasing and can be flat for some intervals of $u\in[u^\mathrm{min},u^\mathrm{max}]$. Indeed, if we consider the fixed day-ahead price, the problem discussed in Remark~\ref{remark1} does not occur (\textit{i.e.}, market does not fail) and we do not need to consider Assumption~\ref{assumption2}. 
 Under this condition, Theorem~\ref{theorem4}~(i) and (ii) hold as before but we should investigate Theorem~\ref{theorem4}~(iii), \textit{i.e.}, when the supply and bid price functions have at least one intersection point. Using the similar analysis in the proof of Theorem~\ref{theorem4}, the existence of the solution to \eqref{eq33} can be shown. However, the uniqueness of the solution to \eqref{eq33} does not necessarily hold. In order to show this,
 consider the pessimistic solution $\ubar{u}:=\arg\inf_{u} \hat{\pi}^{\mathrm{d}}(u|u^{\mathrm{d}}) - f^{\mathrm{d}}(d_0 + \phi(u)) \geq 0$ 
 and optimistic solution $\bar{u}:=\arg\inf_{u} \hat{\pi}^{\mathrm{d}}(u|u^{\mathrm{d}}) - f^{\mathrm{d}}(d_0 + \phi(u))>0$. 
 According to the proofs of Lemma~\ref{lemma3} and Theorem~\ref{theorem4}, the 
 function $\hat{\pi}^{\mathrm{d}}(u^{\mathrm{d}}|u^{\mathrm{d}})$ is monotone increasing and $\frac{\partial f^\mathrm{d}(d_0+\phi(u^{\mathrm{d}}))}{\partial u^{\mathrm{d}}} = -f^{\mathrm{d}^\prime}(d_0+\phi(u^{\mathrm{d}}))\eta(u^{\mathrm{d}})$. Then, if $u^\mathrm{d}\in[\ubar{u},\bar{u}]$ and $\bar{u}\neq\ubar{u}$, we have 
$f^{\mathrm{d}^\prime}(d_0+\phi(u^{\mathrm{d}}))\eta(u^{\mathrm{d}})=0$ for this interval. Thus, since $f^{\mathrm{d}^\prime}(d_0+\phi(u^{\mathrm{d}}))$ is non-negative (\textit{i.e.} it can be equal to zero), $\eta(u^{\mathrm{d}})$ is not necessarily equal to zero (\textit{i.e.}, there can exist some agents) in this interval. Indeed, for each $u^\mathrm{d}\in[\ubar{u},\bar{u}]$, the number of agents which trade on the day-ahead market can be different. Consequently, if we consider the fixed day-ahead price and possible certain congestion occurring, in the Nash equilibrium, there exist a solution characterized by $u^\mathrm{d}$ which is not necessarily unique.

 Moreover, following Theorem~\ref{theorem5},  we can show that the Nash equilibrium for this game exists but it is not necessarily unique.  In order to show this, let $\tilde{u}^\mathrm{d}$ and $\tilde{U}^\mathrm{r}$ satisfy the conditions of Theorem~\ref{theorem2}, 
 regardless of the existence of a Nash equilibrium for this game. Since the day-ahead price is fixed and  congestion occurring can be certain, 
 the solution $\tilde{u}^\mathrm{d}$ to \eqref{eq33} exists but it is not  necessarily unique and can be multiple. Also, let the agents trade on the day-ahead and redispatch market according to Theorem~\ref{theorem2} with $\tilde{u}^\mathrm{d}$ and $\tilde{U}^\mathrm{r}$, regardless of the existence of a Nash equilibrium for this game. Therefore, the agents with  $u_i\leq \tilde{u}^\mathrm{d}$ do not trade on the day-ahead market while  with $\tilde{u}^\mathrm{d}\leq u_i$ trade and their number varies for different solutions $\tilde{u}^\mathrm{d}$. Then, using the analogous analysis in the proof of Theorem~\ref{theorem5}, we can show that for each solution 	$\tilde{u}^\mathrm{d}$, there exists a Nash equilibrium for this game which can be different.
 
 We note that the agents with utilities inside the interval $[\ubar{u},\bar{u}]$ have the zero expected welfare and are indifferent to trade on the day-ahead market or not; however, these agents can affect 
 the DSO costs in the redispatch market. More precisely, for the optimistic solution $u^\mathrm{d}=\bar{u}$, the agents with utilities inside the interval $[\ubar{u},\bar{u}]$ do not trade on the day-ahead market while for the pessimistic solution $u^\mathrm{d}=\ubar{u}$, all of these agents trade. 
 Therefore, 
  trading more agents in the day-ahead market leads to 
   worse congestion and higher DSO costs for redispatching.
 
\section{Simulation Results}\label{sec:sim}
\begin{figure}[t]
	\centering
	\includegraphics[width=\linewidth]{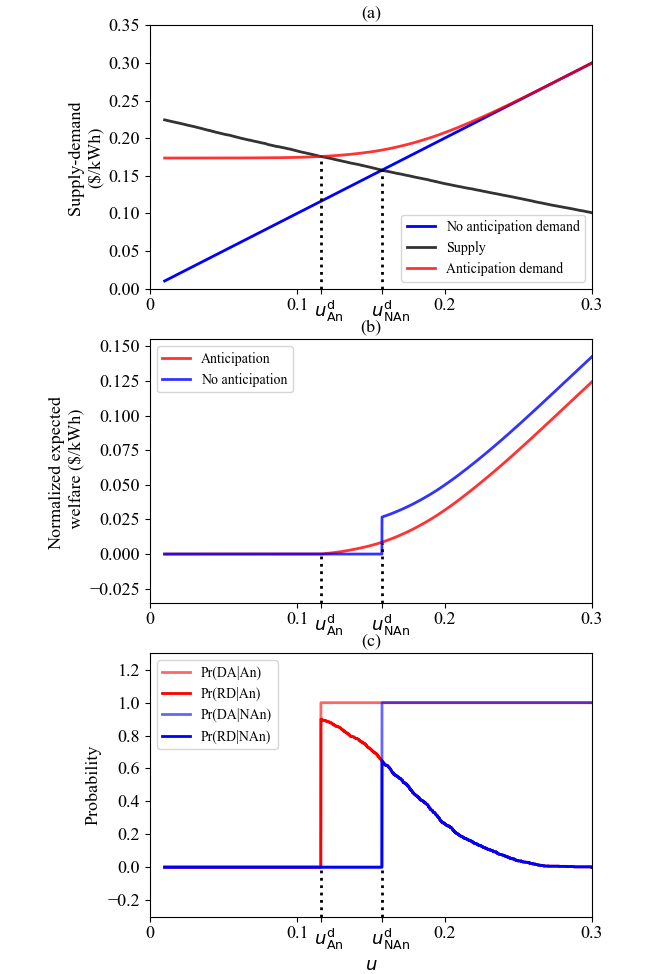}
	\caption{Day-ahead market 
 with medium network capacity and $D\neq 0$: (a) supply and demand curves; (b) normalized expected welfare when the energy consumers can and cannot anticipate the redispatch market; and (c) probabilities of trading on day-ahead and redispatch market.}
	\label{fig2}
\end{figure}

\begin{figure}[t]
	\centering
	\includegraphics[width=\linewidth]{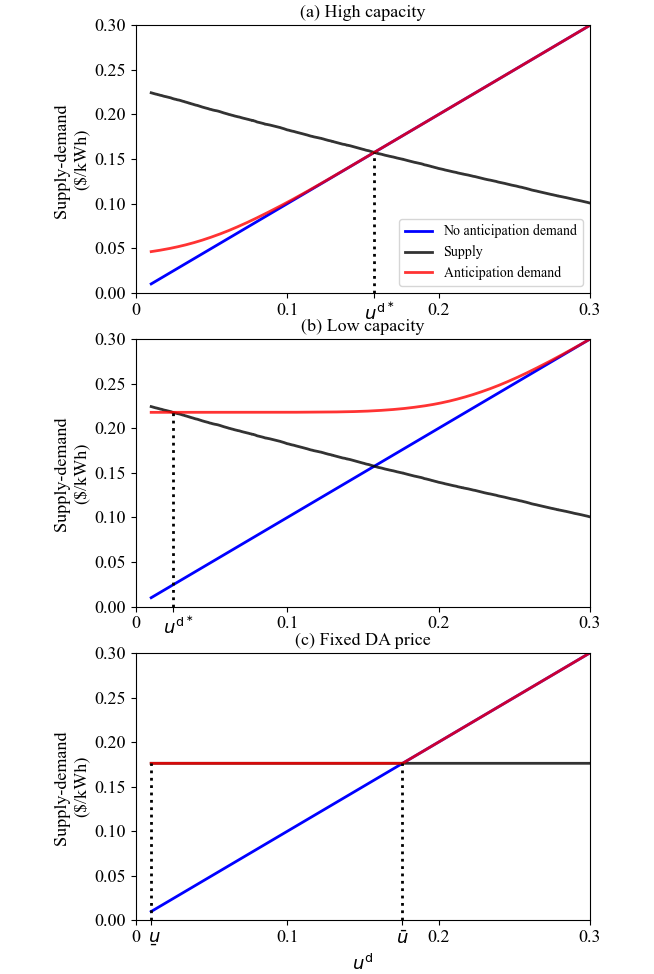}
	\caption{Supply and demand curves:  
 (a) high network capacity; (b) low network capacity; 
 and (c) fixed day-ahead price and no uncertainty.}
	\label{fig3}
\end{figure}

\begin{figure}[t]
	\centering
	\includegraphics[width=\linewidth]{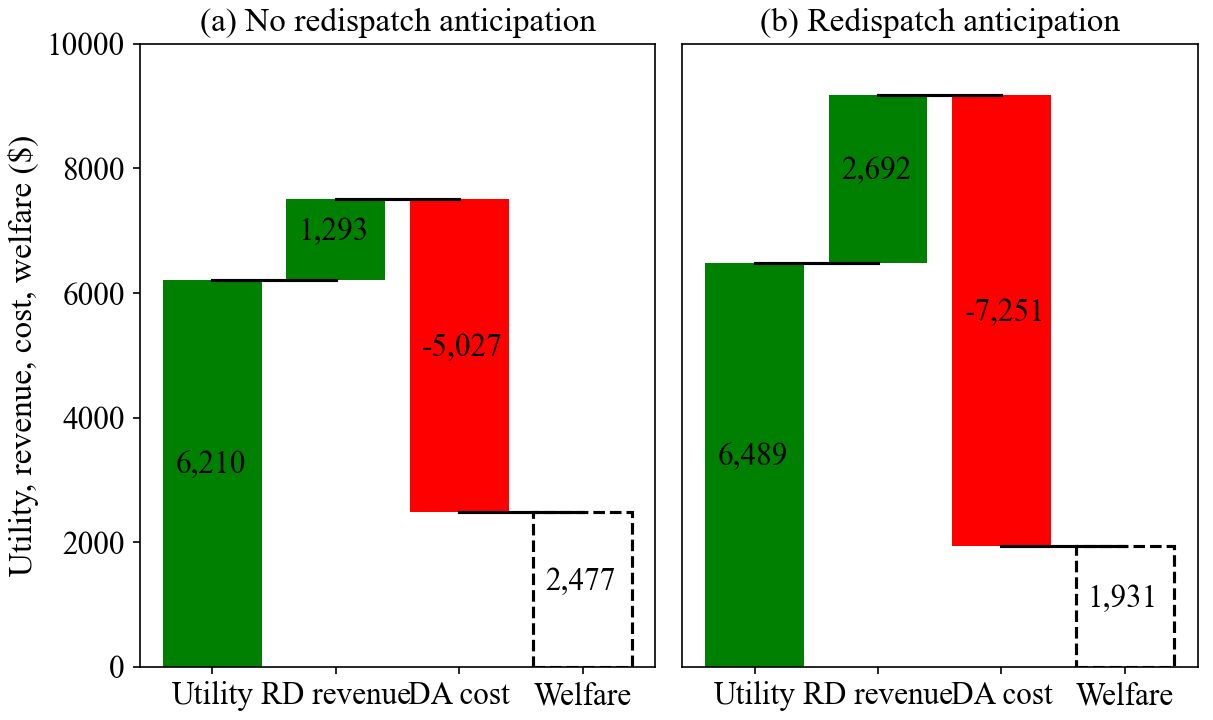}
	\caption{The total expected welfare, utility, redispatch revenue and day-ahead cost at time $t = 1$ with medium network capacity and $D\neq 0$: (a) energy consumers
		cannot anticipate the redispatch market and (b) energy consumers can anticipate the redispatch market.}
	\label{fig10}
\end{figure}

\begin{figure}[t]
	\centering
	\includegraphics[width=\linewidth]{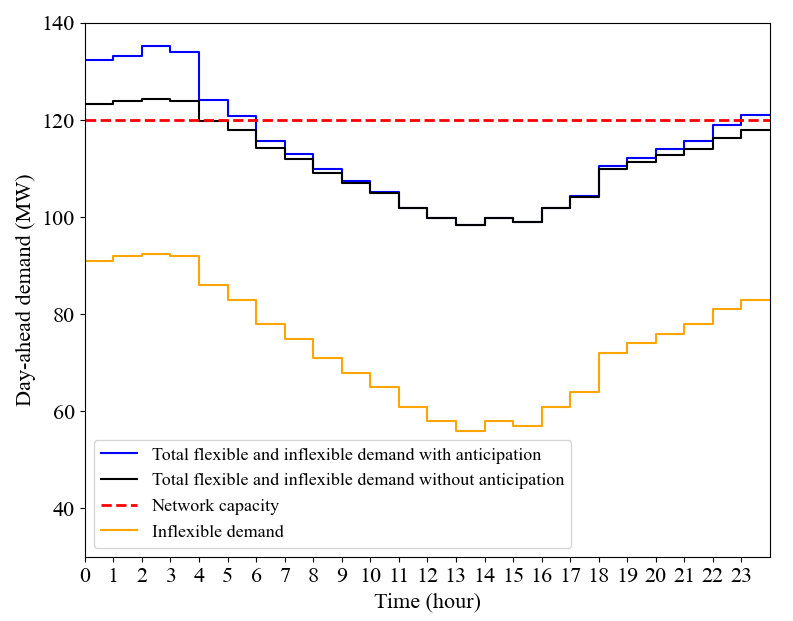}
	\caption{The total normalized expected demand in the day-ahead market when energy consumers can and cannot anticipate the redispatch market.}
	\label{fig4}
\end{figure}

\begin{table*}[t]
	\caption{\label{table1} The total expected welfare, utility, redispatch revenue and day-ahead cost.} 
	\centering
	\begin{tabular}{@{} cccccccc @{}}
 \toprule[0.02cm]
 \textbf{Redispatch anticipation}\\
		\midrule[0.02cm]
      & \textbf{Medium} & \textbf{High}  & \textbf{Low}      & \textbf{Fixed DA price} & \textbf{Fixed DA price} 
  \\	&\textbf{capacity} &\textbf{capacity}  &\textbf{capacity} &  \textbf{(pessimistic)} & \textbf{(optimistic)} \\
		\midrule[0.02cm]
		Utility (\$)~ &$6,489$ &$7,323$ &$4,600$	 		&$6,670$ &$6,670$ \\
		\midrule[0.01cm]
		\begin{tabular}[c]{@{}l@{}l@{}}Redispatch \\revenue (\$) \end{tabular} &$2,692$ &$1$ &$7,588$	 	 &$6,507$ &$0$ \\
		\midrule[0.01cm]
		\begin{tabular}[c]{@{}l@{}l@{}}Day-ahead~ \\ cost (\$)~ \end{tabular}       &$-7,251$ &$-5,029$ &$-11,212$	   &$-11,441$ &$-4,934$  \\
		\midrule[0.01cm]
		\begin{tabular}[c]{@{}l@{}l@{}}Welfare (\$) \end{tabular} &$1,931$ &$2,295$  &$976$   &$1,736$ &$1,736$ \\
  \bottomrule[0.02cm]
      \toprule[0.02cm]
\textbf{~~~No redispatch anticipation}
 \\
		\midrule[0.02cm]
		Utility (\$)~ &$6,210$ &$7,321$ &$4,542$ 		&$6,670$ &$6,670$ \\
		\midrule[0.01cm]
		\begin{tabular}[c]{@{}l@{}l@{}}Redispatch \\revenue (\$) \end{tabular} &$1,293$ &$1$ &$3,383$		 &$0$ &$0$ \\
		\midrule[0.01cm]
		\begin{tabular}[c]{@{}l@{}l@{}}Day-ahead~ \\ cost (\$)~ \end{tabular}       &$-5,027$  &$-5,027$ &$-5,027$   &$-4,934$ &$-4,934$  \\
		\midrule[0.01cm]
		\begin{tabular}[c]{@{}l@{}l@{}}Welfare (\$) \end{tabular} &$2,477$ &$2,295$  &$2,899$   &$1,736$ &$1,736$ \\
  	\bottomrule[0.02cm]
        \toprule[0.02cm]
	\begin{tabular}[c]{@{}l@{}l@{}}Welfare \\difference (\$) \end{tabular} &$-546$ &$0$ &$-1,923$		 &$0$  &$0$  \\
		\bottomrule[0.02cm]
 	\end{tabular}
\end{table*} 

In this section, the performance of the proposed method is verified by numerical simulation. 
We 
use the total inflexible load profile given in \cite{ma}.  The time horizon covers the 24-hour period and we use the supply  function $f^\mathrm{d}(x)=0.15(\frac{x}{1.2\times 10^5})^{1.5}~\$/\text{kWh}$ from \cite{ma,sergio,sergio2}. 
Moreover, the total number of energy consumers is considered $10^4$,  the utility 
and the maximum energy consumption of each energy consumer 
are chosen randomly from  the uniform distributions on the intervals $[0.01,0.3]~\$/\text{kWh}$ and $[3,10]~\text{kWh}$, respectively. We consider the medium network capacity $c=120~\text{MW}$ as the base scenario and also the low and high network capacities $c=110, 150~\text{MW}$, respectively. For the anticipated
forecast inflexible load demand $d_0$, we use the inflexible load profile given in \cite[Figure 1]{ma} and  the stochastic variable of forecast inflexible load demand $D$ is chosen randomly from  the normal distribution with the standard deviation $10~\text{MW}$. 

Now, we consider the two-stage problem in Fig.~\ref{fig1} for this simulation setup. Furthermore, we consider the case that the energy consumers cannot anticipate the redispatch market (naive scenario). Indeed, in this case, the energy consumers do not consider their expected revenues from the redispatch market when they bid in the day-ahead market (they do not bid strategically). However, for computing the expected welfare of the energy consumers, we consider their expected revenues from the redispatch market. Fig.~\ref{fig2} illustrates the supply and demand curves in the day-ahead market, normalized expected welfare  when the energy consumers can and cannot anticipate the redispatch market, and probabilities of trading on day-ahead and redispatch market  
with medium network capacity, respectively. We can observe from Fig.~\ref{fig2}~(b) that the normalized expected welfare of energy consumers, which do not trade on the day-ahead market when the energy consumers cannot anticipate the redispatch market while trade when they can anticipate the redispatch market, increases. However, trading more energy consumers on the day-ahead market results in a higher day-ahead clearing price and the welfare of some energy consumers, which trade on the day-ahead market either when the energy consumers can anticipate the redispatch market or not, decreases. The difference between the curves in Fig.~\ref{fig2}~(b) when the utility is greater than $u^\mathrm{d}_\mathrm{NAn}$, is almost $0.01815~\$/\mathrm{kWh}$. Moreover,  when the energy consumers can anticipate the redispatch market in comparison with the case that they cannot, the energy consumers with utilities inside the interval $[u^\mathrm{d}_\mathrm{An},u^\mathrm{d}_\mathrm{NAn}]$ gain $33 ~\$$ and the energy consumers with utilities greater than $u^\mathrm{d}_\mathrm{NAn}$ lose $579~\$$. We can notice from Fig.~\ref{fig2}~(c) that if the utility of each energy consumers is less than $u^\mathrm{d}_\mathrm{An}$ for the case that they can anticipate or $u^\mathrm{d}_\mathrm{NAn}$ for the case that they cannot anticipate, the probabilities of trading on day-ahead market and redispatch market, is zero; otherwise, they are greater than or equal to zero. Also, the probability of trading on the redispatch market decreases when the utilities increase. 
Fig.~\ref{fig3}  shows the supply and demand curves in the day-ahead market  
with high and low network capacities, medium network capacity with fixed day-ahead price and no uncertainty ($D=0$), respectively. We can note from Figures~\ref{fig3}~(a), (b) and \ref{fig2}~(a) that when the network capacity increases, the value of $u^\mathrm{d \ast}$ increases; therefore, fewer number of energy consumers trade on the day-ahead market and the day-ahead clearing price decreases. For the high network capacity, the probability of congestion occurring is low and the demand curve is very close to the identity line while for the low network capacity, the probability of congestion occurring is high and the demand curve is very close to a constant value.  
We can observe from Fig.~\ref{fig3}~(c) that when the day-ahead price is fixed, 
the pessimistic solution $\ubar{u}$ and optimistic solution $\bar{u}$
exist such that  $\bar{u}\neq\ubar{u}$.
We can note from Fig.~\ref{fig10} that the energy consumers make more redispatch revenue from the DSO when they can anticipate the redispatch market in comparison with the case that they cannot. However, their total welfare does not become necessary higher. This is because when the energy consumers can anticipate the redispatch market, more energy consumers trades on the day-ahead market which leads to a higher day-ahead clearing price in comparison with the case that they cannot anticipate. Indeed, although the welfare of some energy consumers increases when they can anticipate the redispatch market in comparison with the case that they cannot, the welfare of some other energy consumers decreases (see Fig.~\ref{fig2}~(b)). 

Table~\ref{table1} demonstrates the total expected welfare, utility, redispatch revenue and day-ahead cost of the energy consumers  
for different network capacities, fixed day-ahead price and for the cases that the energy consumers can and cannot anticipate the redispatch market.   
We can observe that the total redispatch revenue and day-ahead cost increase when the network capacity decrease (\textit{i.e.}, when the probability of occurring congestion becomes higher).  
Indeed, DSO should pay higher costs for redispatching the energy schedules when the energy consumers anticipate that the congestion problem will occur with a higher probability. 
Furthermore, we can see that for the optimistic solution $u^\mathrm{d}=\bar{u}$,  the energy consumers with utilities inside the interval $[\ubar{u},\bar{u}]$, so-called indifferent agents (see Subsection~\ref{subsec:D} and Fig.~\ref{fig3}~(c)), do not trade on the day-ahead market while for the pessimistic solution $u^\mathrm{d}=\ubar{u}$, all of these energy consumers  trade on the day-ahead market which leads to higher redispatch revenues (DSO costs).
We can notice from Fig.~\ref{fig4} that at times $t=0, 1, 2, 3,4, 5, 23$ the congestion problem is aggravated when the energy consumers can anticipate the redispatch market. Indeed, more energy consumers trades on the day-ahead market at the times when they anticipate the congestion problem may occur to increase their revenues 
in the redispatch market. Therefore, the increase-decrease game aggravates the congestion problem and the DSO should pay high costs for redispatching the schedules of the energy consumers in the redispatch market.

\section{Conclusions and Future Work}\label{sec:con}
In this paper,  we have considered a two-stage problem consisting of the day-ahead market (first stage) and redispatch market (second stage). Then,
we have modeled the increase-decrease game for large populations of energy consumers using a stochastic mean field game 
approach. Next, we have shown that all the agent strategies are ordered along their utilities. Finally, we have proved that a unique Nash equilibrium exists for this game. 
Future works include addressing the increase-decrease gaming using an iterative algorithm based on the mean field reverse Stackelberg game.

\end{document}